\title{Robust Implementation with Costly Information}
\author{Harry Pei\footnote{Department of Economics, Northwestern University. Email: harrydp@northwestern.edu} \and Bruno Strulovici\footnote{Department of Economics, Northwestern University. Email: b-strulovici@northwestern.edu}\footnote{We thank Gabriel Carroll, Matt Jackson, Takashi Kunimoto, Meg Meyer, Stephen Morris, David Rodina,
Satoru Takahashi, and Olivier Tercieux for helpful comments. Pei thanks the NSF Grant SES-1947021 for financial support.}}
\date{\today}
\documentclass[11pt]{article}
\usepackage{amsmath}
\usepackage{graphicx}
\usepackage{amsfonts}
\usepackage{amsthm}
\usepackage{setspace}
\usepackage{tikz}
\usetikzlibrary{patterns}
\usepackage{sgame}
\usepackage{color}
\usepackage{hyperref}
\usepackage[top=1in, bottom=1in, left=1in, right=1in]{geometry}
\begin{document}

\numberwithin{equation}{section}

\newtheorem{Proposition}{\hskip\parindent\bf{Proposition}}
\newtheorem{Theorem}{\hskip\parindent\bf{Theorem}}
\newtheorem*{Theorem1}{\hskip\parindent\bf{Theorem 1'}}
\newtheorem*{Theorem3}{\hskip\parindent\bf{Theorem 3'}}
\newtheorem*{Assumption2}{\hskip\parindent\bf{Assumption 2'}}
\newtheorem{Lemma}{\hskip\parindent\bf{Lemma}}
\newtheorem{Corollary}{\hskip\parindent\bf{Corollary}}
\newtheorem{Definition}{\hskip\parindent\bf{Definition}}
\newtheorem{Assumption}{\hskip\parindent\bf{Assumption}}
\newtheorem{Condition}{\hskip\parindent\bf{Condition}}
\newtheorem{Conjecture}{\hskip\parindent\bf{Conjecture}}
\newtheorem{Claim}{\hskip\parindent\bf{Claim}}

\maketitle
\noindent \textbf{Abstract:} We study whether a planner can robustly implement a state-contingent social choice function when (i) agents must incur a cost to learn the state and (ii) the planner faces uncertainty regarding agents' preferences over outcomes, information costs, and beliefs and higher-order beliefs about one another's payoffs. We propose mechanisms that can approximately implement any desired social choice function when the perturbations concerning agents' payoffs have small ex ante probability. The mechanism is also robust
to trembles in agents' strategies and when agents receive noisy information about the state.\\

\noindent \textbf{Keywords:} Robust Implementation, Partial Implementation, Critical Path Lemma.

\begin{spacing}{1.5}
\section{Introduction}\label{sec1}
Implementation theory studies whether and how a state-contingent social choice function, such as convicting guilty defendants and acquitting innocent ones, can be achieved when (i) the information necessary to implement the objective is unknown to the planner and (ii) the social choice function is in conflict with the interests of the agents who do have access to the relevant information.

We study this question under two additional constraints: (iii) agents need to incur a cost to learn the relevant information and (iv) the planner knows agents' payoffs with probability close to one, but faces uncertainty regarding agents' payoffs with the remaining probability and regarding agents' beliefs and higher-order beliefs about one another's payoffs.

The motivation for our research question is twofold. First, in many situations of interest, agents do not possess the relevant information at the outset and must acquire it at some cost. For example, investigators need to exert effort in order to learn whether a defendant is guilty or innocent.

Second, the literature on robust implementation, pioneered by Bergemann and Morris (2005), has underscored the importance of implementing social choice functions when agents' preferences are not common knowledge. This literature has shown that only a restrictive subset of social choice functions can be implemented when robustness is required to hold {\em globally}: a social choice function can be implemented for arbitrary payoffs and beliefs of the agents only if it is ex post incentive compatible (Bergemann and Morris 2005). Even when the notion of implementation is {\em local} but in an {\em interim} sense, i.e., for all profiles of agent-types close to a given profile, implementable social choice functions must satisfy {\em Maskin monotonicity} (Oury and Tercieux 2012), a demanding property that is often violated in applications.

We consider a novel notion of robust implementation that builds on the concept of equilibrium robustness introduced by Kajii and Morris (1997). A Nash equilibrium of a complete information game is robust if it can be approximated by some equilibria in all nearby incomplete information games, which are all games in which players' payoffs match those of the complete information game with probability close to one. Building on this approach, our concept of robust implementation is \textit{local} and \textit{ex ante}, rather than \textit{interim}, in the sense that the perturbations considered have a small probability ex ante relative to the complete information game. As a result, the concept has the potential to avoid some of the most stringent implications of global and interim concepts.



Our concept of robust implementation departs from Kajii and Morris (1997) by imposing a key restriction on the perturbations considered by the planner. Since we study a mechanism design problem rather than an arbitrary game, we focus on perturbations in which agents' payoffs do not depend per se on the {\em messages} that they send to the mechanism. Instead, perturbations pertain to agents' preferences concerning the {\em outcomes} implemented as a result of their messages and it is common knowledge that messages are cheap talk.

We show that under our notion of robust implementation, every social choice function is robustly implementable.\footnote{This result holds under a generic assumption on the objective state distribution, which is that one state is a priori more likely than any other states. This generic assumption can be dropped when there is a known upper bound on agents' costs of learning.} Our analysis focuses on the case of two agents.\footnote{Our main result is a possibility result: we construct a mechanism that implements any desired social choice function when two agents have the ability to learn the state. This mechanism can straightforwardly be extended to any arbitrary number agents, for instance by applying it to two specific agents and ignoring the reports of all remaining agents.} A planner wishes to implement a social choice function that maps a finite set of states to a finite set of outcomes. The planner commits to a mechanism mapping agents' messages to outcomes and transfers. Agents observe the mechanism as well as their own payoff functions under the perturbation. They independently decide whether to observe the state at some cost and then send messages to the planner.\footnote{In the spirit of our robustness exercise, we allow small amount of uncertainty regarding the accuracy of agents' signal, and allow agents to tremble with small probability when sending messages. See Section 4.3 for details.}

We construct a class of mechanisms called \textit{Augmented Status Quo Rules} that robustly implements any social choice function. These mechanisms treat states asymmetrically by (i) introducing a status quo outcome, implemented whenever agents reports are in conflict or when they match on the state that is the most likely ex ante, and (ii) rewarding agents differently when their reports match depending on the messages they send. 

In addition, these mechanisms include {\em strictly more messages than states}: one message for the ex ante most likely state and two messages for every other state. The two messages which represent the same state induce the same outcome regardless of the other agent's report but lead to different transfers. In Section \ref{sec2}, we use an example to explain why this larger message space makes the mechanism robust to uncertainty with respect to agents' biases over outcomes, agents' information costs, agents' trembles, and to noise in agents' private signals about the state.

Finally, we provide several results concerning robust implementation for stronger notions of robustness. First, Proposition~\ref{Prop1} shows that it is impossible to approximate implement any non-constant social choice function if the probability of agents' payoff perturbations is bounded away from zero. This result illustrates the role for implementability of our notion's focus on ex ante \textit{local} perturbations, i.e., perturbations for which agents' payoffs coincide with the unperturbed environment with probability close to one.

Second, our notion of robustness concerns {\em partial} implementation: we only require that the desired social choice function be implemented by some (not necessarily all) equilibria of the game induced by our mechanism. In fact, when agents' unperturbed preferences put no weight on how outcomes relate to the state of the world, or when agents' costs of learning in the unperturbed environment are above some cutoff, we show (Proposition~\ref{Prop2}) that no mechanism can fully implement, even approximately implement, any non-constant social choice function. In these situations, there always exists an equilibrium in which no agent learns the state.\footnote{This result echoes Strulovici (2021), who shows in a sequential model of learning that when agents' preferences are state independent, implementation is impossible even in a partial sense when signals about the state of the world are subject to an {\em information attrition} condition.} We also show a feasibility result under a strong condition relating agents' preferences and the social choice function: when at least one agent cares directly about the state of the world and his preference and the social choice function satisfy a strict version of the cyclical monotonicity condition introduced by Rochet (1987), the planner can robustly and fully implement that social choice function by ignoring the report of the other agent.

Section~\ref{sec2} illustrates our research question and main results by a simple example featuring a binary state of the world. We explain why mechanisms that treat states symmetrically by rewarding agents by a fixed amount when their reports match and giving them nothing when their reports are different, and uniformly randomizing across social choices when agents' reports mismatch, fail to implement desired social choice functions. We then introduce mechanisms that treat states asymmetrically by singling out one state to determine a status quo outcome and by providing asymmetric rewards across states, and explain how this asymmetry helps address robustness issues. We also explain using strictly more messages than states helps achieve robustness when agents' cost of information are uncertain. The general model is introduced in Section~\ref{sec3} and our main results are stated and derived in Section~\ref{sec4}. Finally, Section~\ref{sec5} presents impossibility results pertaining to global and full notions of implementation, and one possibility result for full implementation when some agents have specific state-sensitive preferences.

\section{Example}\label{sec2}
Consider a social planner facing a defendant who may be guilty or innocent. The defendant's prior probability of guilt is~$q \in (0,1)$.

The planner's objective is to convict guilty defendants and acquit innocent ones. She commits to a mechanism $\mathcal{M} \equiv \{M_1,M_2,g,t_1,t_2\}$ in order to solicit information from two agents, where $M_i$ is a finite set of messages for agent $i$, $g: M_1 \times M_2 \rightarrow [0,1]$ is a mapping from messages to the probability of conviction, and $t_i: M_1 \times M_2 \rightarrow \mathbb{R}$ is the transfer to agent $i$.

Each agent can, at some cost,  conduct an investigation and learn whether the defendant is guilty or innocent. Agents' decisions to investigate and the observations resulting from these investigations are private.

We assume for now that agents' payoffs are common knowledge and symmetric: agent~$i$'s realized payoff is $t_i-c \chi_i$ where $\chi_i \in \{0,1\}$ denotes $i$'s decision of whether to conduct investigation and $c$ is $i$'s cost of conducting his investigation.

The planner's objective of convicting guilty defendants and acquitting innocent ones can be achieved  via the following Maskin mechanism, in which each agent is asked to report whether the defendant is guilty or innocent. The outcome and the transfers are given by:
\begin{center}
\begin{tabular}{| c | c | c |}
\hline
  transfers & innocent & guilty  \\
  \hline
  innocent & $R,R$ & $0,0$ \\
  \hline
  guilty & $0,0$ & $R,R$  \\
  \hline
\end{tabular}
\begin{tabular}{| c | c | c |}
\hline
 outcome & innocent & guilty  \\
  \hline
  innocent & acquit & convict with prob $1/2$ \\
  \hline
  guilty & convict with prob $1/2$ & convict  \\
  \hline
\end{tabular}
\end{center}
As long as the reward $R>0$ is large enough relative to the cost $c$, there exists an equilibrium where both agents
conduct investigations and report their findings truthfully.

\paragraph{Challenges with Biased Agents:}
The previous mechanism fails to achieve the objective if agents are subject to biases over outcomes, even if these biases occur with arbitrarily small probability. To fix see this, consider the following perturbation of agents' payoffs and information. Nature draws a ``circumstance'' $\omega$ from the space $\Omega \equiv \{\omega_0,\omega_1,\omega_2,...\}$ such that $\omega=\omega_t$ occurs with probability $\eta(1-\eta)^t$ for every $t \in \mathbb{N}$ where $\eta>0$ is close to $0$. The realization of $\omega$ is independent of whether the defendant is guilty or innocent. Agent 2's payoff is $t_2-\chi_2 c$ at every $\omega \in \Omega$. Agent $1$'s payoff is $t_1-\chi_1 c$ at every $\omega \in \Omega \backslash \{\omega_0\}$, but includes a large benefit $B>0$ from acquitting the defendant if $\omega=\omega_0$ (e.g., agent 1 is the defendant's friend when $\omega=\omega_0$). Agent $1$ knows which element of the partition $\{\omega_0\},\{\omega_1,\omega_2\},\{\omega_3,\omega_4\},...$ the realized $\omega$ belongs to before deciding whether to conduct his investigation as well as what to report.
Agent $2$ knows which element of the partition $\{\omega_0,\omega_1\},\{\omega_2,\omega_3\},...$ the realized $\omega$ belongs to before deciding whether to conduct his investigation as well as what to report. Each agent updates his belief about the other agent's type according to Bayes Rule.\footnote{For example, the type of agent $2$ who knows that $\omega \in \{\omega_0,\omega_1\}$ attaches probability $\frac{1}{2-\eta}$ to agent 1 being type $\{\omega_0\}$, the type of agent $1$ who knows that $\omega \in \{\omega_1,\omega_2\}$ attaches probability $\frac{1}{2-\eta}$ to agent 2 being type $\{\omega_0,\omega_1\}$.}

Under this perturbation, the Maskin mechanism fails to implement the desired objective if $B$ is too large: for every $R \in \mathbb{R}_+$, there exists $B > R$ such that no matter how small~$\eta$ is, there is a unique equilibrium of the perturbed environment and this equilibrium is such that agents always report that the defendant innocent and the defendant is acquitted regardless of his guilt. To understand this result, notice that when $\omega=\omega_0$, agent $1$ wants to maximize the probability of acquittal, so he \textit{reports innocent} with probability~1. When $\omega \in \{\omega_0,\omega_1\}$, agent $2$ is unbiased, but he believes that agent~1 is biased with probability greater than $1/2$, so he believes that agent 1 will \textit{report innocent} with probability greater than $1/2$. Since agent 2 maximizes his expected transfer minus his cost of investigation, he has a strict incentive to \textit{report innocent}. By induction, this contagion argument shows that both agents report that the defendant is innocent with probability one in the unique equilibrium of this perturbed environment.

In general, agents may be biased in either direction (they may also receive large benefits from convicting the defendant) and with arbitrary magnitude. The planner faces uncertainty about the direction and magnitude of these biases as well as about agents' beliefs and higher-order beliefs about each other's biases. The planner aims to design a mechanism that approximately implements the desired objective
when agents are unbiased with probability close to $1$ but can have arbitrary biases with small but positive probability and may have arbitrary beliefs and higher-order-beliefs as long as those beliefs can be derived from a common prior.

\paragraph{Biased Uncertainty \& Status Quo Rule with Ascending Transfers:}

We now propose a mechanism that achieves this objective, maintaining for now the assumption that agents' cost of acquiring information is fixed and commonly known. We call this mechanism \textit{Status Quo Rule with Ascending Transfers}. This mechanism still features two messages for each agent: \textit{innocent} and \textit{guilty}. The outcome and transfers are specified as followss:
\begin{center}
\begin{tabular}{| c | c | c |}
\hline
outcome & innocent & guilty  \\
  \hline
  innocent & acquit & acquit \\
  \hline
  guilty & acquit & {\color{blue}{convict}}  \\
  \hline
\end{tabular}
\begin{tabular}{| c | c | c |}
\hline
transfers & innocent & guilty  \\
  \hline
  innocent & $R^1,R^1$ & $0,0$ \\
  \hline
  guilty & $0,0$ & {\color{blue}{$R^2,R^2$}}  \\
  \hline
\end{tabular}
\end{center}
where $R^2-R^1$ is strictly positive and bounded below by a function of $c$.

This mechanism features a status quo outcome, \textit{acquit}, which is implemented as long as one agent reports \textit{innocent}. The defendant is convicted if and only if both agents report \textit{guilty}. Agents receive positive transfers only when their reports coincide and they receive a larger transfer when both of them report \textit{guilty}.

To explain intuitively why this mechanism implements the desired objective, we restrict attention to the following class of perturbations and defer the general proof to Section \ref{sec4}. Nature draws $\omega$ from a countable set $\{\omega_0,\omega_1,\omega_2,...\}$ according to distribution $\Pi \in \Delta \{\omega_0,\omega_1,\omega_2,...\}$. Agent $1$'s information partition is $\{\omega_0\},\{\omega_1,\omega_2\},\{\omega_3,\omega_4\},...$ Agent 2's information partition is $\{\omega_0,\omega_1\},\{\omega_2,\omega_3\},...$ Agent 2's payoff is $t_2 - c \chi_2$ at every $\omega$. Agent 1's payoff is $t_1-c \chi_1$ at every $\omega$ except for $\omega_0$.

Each perturbation in the above class is characterized by $\Pi$ and agent 1's payoff function at $\omega_0$, under which we construct an equilibrium that approximately implements the desired outcome.
\begin{enumerate}
  \item Suppose first that when $\omega=\omega_0$, agent $1$ receives an arbitrarily large benefit from {\bf acquitting} the defendant. This biased type can guarantee his desired outcome by reporting \textit{innocent} no matter what. However, given that $R^1<R^2$, $\Pi(\omega_1)$ needs to be strictly less than $\Pi(\omega_0)$ for type $\{\omega_0,\omega_1\}$ of agent 2 to report \textit{innocent} no matter what, and $\Pi(\omega_2)$ needs to be strictly less than $\Pi(\omega_1)$ for type $\{\omega_1,\omega_2\}$ of agent 1 to report \textit{innocent} no matter what,... The upper bounds on these probabilities form a decaying geometric sequence, which means that the contagion caused by such a biased type is bounded from above by a linear function of $\Pi(\omega_0)$.
  \item Suppose now that when $\omega=\omega_0$, agent $1$ receives an arbitrarily large benefit from {\bf convicting} the defendant. If
  $\Pi(\omega_t)=\eta(1-\eta)^t$ for every $t \in \mathbb{N}$ and
  this biased type \textit{reports guilty no matter what}, then all types of both agents will \textit{report guilty no matter what}.
      However, under the outcome function of our mechanism,  the defendant is convicted only if \textit{both agents report guilty}, so the biased type has no incentive to report \textit{guilty} if the defendant is innocent and he believes that all other types report truthfully.
  In fact, this biased type has a strict incentive to be truthful under such a belief since
  when the defendant is innocent,
  reporting \textit{innocent}  leads to a strictly positive transfer and reporting guilty leads to zero transfer. Hence, all types report truthfully is an equilibrium no matter how large the bias is.
\end{enumerate}

\paragraph{Uncertainty about Cost \& Augmented Status Quo Rule:} In general, uncertainty about agents' preferences may also involve uncertainty about agents' costs of information, correlation between these costs and agents' biases, and agents' beliefs and higher-order-beliefs about each other's costs and biases.

We now describe a mechanism that approximately implements the desired objective when agents are unbiased and have cost $c$ with probability close to $1$ but can have arbitrary biases and costs with small probability.

Our mechanism robustly implements the desired objective as long as the prior probability of guilty of the defendant is not equal to $\frac{1}{2}$.\footnote{In the general setting of Section~\ref{sub4.2}, we assume that there is a state whose prior probability of occurrence exceeds that of every other state. This property is generic among the set of all prior distributions.} To fix ideas, we assume that the defendant's probability of guilty is strictly less than $\frac{1}{2}$ and we focus for expositional simplicity on the class of information structures considered before: agent 1's information partition is
$\{\omega_0\},\{\omega_1,\omega_2\},\{\omega_3,\omega_4\},...$ and agent 2's information partition is $\{\omega_0,\omega_1\},\{\omega_2,\omega_3\},...$ We further assume that agent $2$'s payoff function is $t_2-c\chi_2$ at every $\omega$ and that agent $1$'s payoff function is $t_1-c\chi_1$ at every $\omega$ except at $\omega_0$, where it is given by $u_1(\theta,y)+t_1-\widetilde{c} \chi_1$ for some arbitrary function $u_1(\theta,y)$ and cost $\widetilde{c}$.

We start by explaining why the status quo rule with ascending transfers \textit{cannot} implement the desired outcome if agent 1's payoff at $\omega_0$ includes a large benefit from convicting the defendant and a very large cost from acquiring information. Intuitively, when such a high-cost biased type believes that the other agent reports truthfully, he prefers to report \textit{guilty} conditional on the defendant being guilty, since he benefits from convicting the defendant.
If this type reports \textit{innocent} when the defendant is innocent, then he needs to conduct an investigation, but his cost of doing so outweighs his benefit from the transfers. This explains why the high-cost biased type prefers to \textit{report guilty no matter what}. This causes contagion when $\Pi(\omega_t)=\eta(1-\eta)^t$ for every $t \in \mathbb{N}$ no matter how small $\eta$ is.

We now propose a mechanism called the \textit{Augmented Status Quo Rule with Ascending Transfers}, in which each agent has a third message available, which we denote $-$\textit{guilty}. The outcome and transfers of this mechanism are given by:
\begin{center}
\begin{tabular}{| c | c | c | c |}
\hline
  outcome & $-$guilty & innocent & guilty \\
  \hline
  $-$guilty & {\color{blue}{convict}} & acquit & {\color{blue}{convict}}\\
  \hline
  innocent & acquit & acquit & acquit \\
  \hline
  guilty &  {\color{blue}{convict}} & acquit & {\color{blue}{convict}} \\
  \hline
\end{tabular}
\begin{tabular}{| c | c | c | c |}
\hline
  transfers & $-$guilty & innocent & guilty \\
  \hline
  $-$guilty & $R^0, R^0$ & $R^0,R^0$ & $0,0$\\
  \hline
  innocent & $R^0,R^0$ & ${\color{blue}{R^1,R^1}}$ & $0,0$ \\
  \hline
  guilty & $0,0$ & $0,0$ & ${\color{blue}{R^2,R^2}}$ \\
  \hline
\end{tabular}
\end{center}
where $\frac{R^0}{R^2} \approx 1$ and $R^2-R^1$ and $R^1-R^0$ are positive and bounded below by some function of $c$.

We note that the message \textit{$-$guilty} implements the same outcome as message \textit{guilty} regardless of the other agent's message. Moreover, each agent can unilaterally implement the status quo outcome \textit{acquit} by reporting \textit{innocent}. Finally, coordinating on the message \textit{$-$guilty} leads to a lower transfer compared to coordinating on any of the other two messages.

To understand why the message \textit{$-$guilty} makes our mechanism robust to biased types that have high learning costs, suppose that every non-biased type never reports \textit{guilty} when the defendant is innocent (but could report $-$\textit{guilty} or \textit{innocent} with arbitrary probabilities). then, agent 1's type $\omega_0$ receives an expected transfer of at least  $(1-q)R^0$ if he reports \textit{$-$guilty} and an expected transfer of at most $q R^2$ if he reports \textit{guilty}, where $q$ is the ex ante probability that the defendant is guilty. Given our assumptions that the defendant is innocent with probability above $\frac{1}{2}$ (i.e., $q<\frac{1}{2}$) and that $\frac{R^0}{R^2} \approx 1$, reporting \textit{$-$guilty} leads to a higher expected transfer.

Moreover, the assumption that $0<R^0<R^1<R^2$ implies that coordinating on message $-$\textit{guilty} leads to a lower transfer compared to coordinating on message \textit{innocent} and coordinating on message \textit{guilty}. As a result, every type of agent who is unbiased and whose cost is $c$ strictly prefers to conduct the costly investigation and report his finding truthfully as long as his belief satisfies the following two conditions: (1) no type of the other agent reports \textit{guilty} when the defendant is innocent, and (2) the other agent is truthful with probability at least $\frac{1}{2}$.

\section{General Model}\label{sec3}
\paragraph{Primitives:} A planner wants to implement a social choice function $f: \Theta \rightarrow \Delta (Y)$ where
$\Theta$ is a finite set of states and $Y$ is a finite set of outcomes. The typical elements in these sets are $\theta \in \Theta$ and $y \in Y$. Let $q(\theta)$ be the probability of state $\theta$. We assume that $q(\theta)>0$ for every $\theta \in \Theta$.

The planner does not know $\theta$ and elicits information about $\theta$ from two agents.
She commits to a mechanism $\mathcal{M} \equiv \{M_1,M_2,t_1,t_2,g\}$, where $M_i$ is a finite set of messages for agent $i \in \{1,2\}$, $t_i: M_1 \times M_2 \rightarrow \mathbb{R}$ is the transfer to agent $i$, and $g: M_1 \times M_2 \rightarrow \Delta (Y)$ is the implemented outcome.

After observing $\mathcal{M}$, agents simultaneously and independently decide whether to observe $\theta$ at some costs.
Let $\chi_i \in \{0,1\}$ be agent $i$'s decision to obtain information, where $\chi_i=1$ represents agent $i$ obtaining information about $\theta$ and vice versa. Let
$c_i \geq 0$ be his cost.
We assume that information acquisition is \textit{covert} in the sense that neither agent $-i$ nor the planner can observe $\chi_i$.

Then the agents simultaneously send messages $(m_1,m_2) \in M_1 \times M_2$ to the planner, after which the planner makes transfers and implements an outcome according to $\mathcal{M}$. Agent $i$'s payoff is:
\begin{equation}\label{2.1}
    t_i-\chi_i c_i+u_i(\theta,y).
\end{equation}
Following Lipnowski and Ravid (2020), we say that agent $i$ has \textit{transparent motives} if $u_i(\theta,y)$ does not depend on $\theta$.

One assumption to highlight is that the agents' transfers \textit{cannot} depend on the realized state. This stands in contrast to existing works on contracting with costly information acquisition, such as Zermeno (2011), Carroll (2019), and Clark and Reggiani (2021) where transfers can also depend on the realized state. Our model fits situations where either the principal cannot verify the state ex post, or additional information about the state takes a long time to arrive so that rewarding agents based on such late information is impractical.

\paragraph{Robustness Concerns:} We examine whether the planner can \textit{robustly} implement $f$ when agents' preferences over outcomes, costs of obtaining information, and their beliefs and higher-order-beliefs about each other's payoffs can be different from what the planner believed to be.
We focus on partial implementation: the planner only requires $f$ to be implemented in \textit{one} equilibrium.

Following Kajii and Morris (1997), a \textit{perturbation}
\begin{equation}\label{2.2}
\mathcal{G} \equiv \{\Omega,\Pi,Q_1,Q_2,\widetilde{u}_1,\widetilde{u}_2,\widetilde{c}_1,\widetilde{c}_2\}
\end{equation}
consists of  a countable set of \textit{circumstances} $\Omega$, a distribution $\Pi \in \Delta (\Omega)$ over the set of circumstances which we assume is independent of $\theta$, a partition $Q_i$ of $\Omega$ such that agent $i \in \{1,2\}$ knows
which element of the partition $Q_i$ the realized $\omega$ belongs to, and mappings
$\widetilde{u}_i: \Omega \times \Theta \times Y \rightarrow \mathbb{R}$ and $\widetilde{c}_i: \Omega \rightarrow \mathbb{R}_+$ for $i \in \{1,2\}$.
Agent $i$'s payoff under perturbation $\mathcal{G}$ is
\begin{equation}\label{2.3}
    t_i-\widetilde{c}_i (\omega) \chi_i + \widetilde{u}_i (\omega,\theta,y).
\end{equation}
For every $\overline{c} \geq \max \{c_1,c_2\}$, $\mathcal{G}$ is a $\overline{c}$-bounded perturbation if $\widetilde{c}_i(\omega) \leq \overline{c}$ for every $i$ and $\omega \in \Omega$.

For every $\omega \in \Omega$, let $Q_i(\omega)$ be the partition element of $Q_i$ that contains $\omega$, which we call agent $i$'s \textit{type}.
Type $Q_i(\omega)$ is a \textit{normal type} if $\widetilde{u}_i(\omega',\theta,y)=u_i(\theta,y)$ and $\widetilde{c}_i(\omega')=c_i$ for every $\omega' \in Q_i(\omega)$, i.e., this type of agent $i$ knows that his payoff in the perturbed environment coincides with his payoff in the unperturbed environment.
For every $\eta \in [0,1]$, we say that $\mathcal{G}$ is an \textit{$\eta$-perturbation} if
\begin{equation}\label{2.4}
    \Pi \Big(\textrm{both agents are normal types}\Big) \geq 1-\eta.
\end{equation}
The class of perturbations considered in our leading example are $\eta$-perturbations since both agents are normal types when $\omega \in \Omega \backslash \{\omega_0\}$, and event $\Omega \backslash \{\omega_0\}$ occurs with probability $1-\eta$.

Intuitively, a perturbation is \textit{small} if agents' payoffs coincide with those in the unperturbed environment with probability close to one, but their payoffs can be very different from those in the unperturbed environment with low but positive probability.
Although a normal type's payoff coincides with that in  the unperturbed environment,
he may believe that the other agent is not normal, and may believe that the other agent thinks that he is not normal, and so on.

A perturbation $\mathcal{G}$ and a mechanism $\mathcal{M} \equiv \{M_1,M_2,g,t_1,t_2\}$ induce an incomplete information game between the two agents, which we denote by $(\mathcal{M},\mathcal{G})$. A typical strategy profile is denoted by $\sigma$.
Let $g_{\sigma}(\theta) \in \Delta (Y)$ be the implemented outcome conditional on the state being $\theta$ when the planner commits to outcome function $g$ and agents behave according to strategy profile $\sigma$.

We introduce two notions of \textit{local robust implementation}, that is, the planner designs a mechanism that approximately implements $f$ for \textit{all} small enough perturbations.
\begin{enumerate}
  \item  We say that $\mathcal{M}$ \textit{robustly implements} $f$ if for every $\varepsilon>0$, there exists $\eta>0$ such that for every $\eta$-perturbation $\mathcal{G}$, there exists an equilibrium $\sigma(\mathcal{G})$ of the incomplete information game induced by $(\mathcal{M},\mathcal{G})$,
such that
\begin{equation}\label{2.5}
\max_{\theta \in \Theta} ||    g_{\sigma(\mathcal{G})}(\theta)-f(\theta) ||_{\textrm{TV}} < \varepsilon,
\end{equation}
where $||\cdot||_{TV}$ is the total variation distance between two distributions.
  \item We say that $\mathcal{M}$ \textit{robustly implements $f$} for all \textit{$\overline{c}$-bounded perturbations} if for every $\varepsilon>0$, there exists $\eta>0$ such that for every $\overline{c}$-bounded $\eta$-perturbation $\mathcal{G}$, there exists an equilibrium $\sigma(\mathcal{G})$ of the incomplete information game induced by $(\mathcal{M},\mathcal{G})$ such that inequality (\ref{2.5}) holds.
\end{enumerate}
By definition, our two notions of robust implementation differ only in terms of whether we allow for unbounded costs of learning in perturbed environments.
If $\mathcal{M}$ robustly implements $f$, then it robustly implements $f$ for all $\overline{c}$-bounded perturbations for every $\overline{c} \geq \max\{c_1,c_2\}$.

We also consider two notions of \textit{global robust implementation}, namely, the planner designs a mechanism $\mathcal{M}$ that approximately implements $f$ under all perturbations, including those where agents' payoffs are different from their payoffs in the unperturbed environment with high probability.
\begin{enumerate}
  \item Mechanism $\mathcal{M}$ \textit{globally implements} $f$ if for every $\varepsilon>0$ and every perturbation $\mathcal{G}$, there exists an equilibrium $\sigma(\mathcal{G})$ of incomplete information game $(\mathcal{M},\mathcal{G})$ such that inequality (\ref{2.5}) holds.
  \item Mechanism $\mathcal{M}$ \textit{globally implements $f$ for all $\overline{c}$-bounded perturbations} if for every $\varepsilon>0$ and every $\overline{c}$-bounded perturbation $\mathcal{G}$, there exists an equilibrium $\sigma(\mathcal{G})$ of incomplete information game $(\mathcal{M},\mathcal{G})$ such that inequality (\ref{2.5}) holds.
\end{enumerate}
By definition, if $\mathcal{M}$ globally implements $f$ (for all $\overline{c}$-bounded perturbations), then it robustly implements $f$
(for all $\overline{c}$-bounded perturbations). In Section \ref{sub4.4}, we show that for any $\overline{c} \geq \max\{c_1,c_2\}$ and any social choice function $f$ that depends non-trivially on the state, there is no finite mechanism that can globally implement $f$ for all $\overline{c}$-bounded perturbations.

\paragraph{Remark:} Our formulation restricts attention to perturbations where agents have quasi-linear utility functions and their payoffs do not directly depend on their messages (which are their actions in our model).\footnote{Although our baseline model restricts attention to quasi-linear payoff functions, our main result generalizes to situations where agents' utilities from transfers are $v_1(t_1)$ and $v_2(t_2)$ in the unperturbed environment, and are $\widetilde{v}_1(\omega,t_1)$ and $\widetilde{v}_2(\omega,t_2)$ in the perturbed environment, as long as there exists a positive constant $\delta$ such that $v_1'(t_1),v_2'(t_2), \frac{\partial \widetilde{v}_1(\omega,t_1)}{\partial t_1}, \frac{\partial \widetilde{v}_2(\omega,t_2)}{\partial t_2}\geq \delta$ for all $\omega,t_1,t_2$. That is, it is common knowledge that agents' marginal utilities from transfers is bounded from below.} Both assumptions are common in the mechanism design literature, including Rochet (1987), Bergemann and Morris (2005, 2009), and Chung and Ely (2007).

Our assumptions stand in contrast to models on robust prediction in games such as Kajii and Morris (1997) and Ui (2001) where players' actions can directly affect their payoffs. Since we consider a mechanism design setting, agents' message spaces are endogenously chosen by the planner, so these messages have no meaning per se and can be viewed as cheap talk.
In many applications, it is also reasonable to assume that all types of the agent prefer more transfers, while agents' preferences over other aspects of the allocation (e.g., whether to convict or acquit the defendant) are more subtle and may not be known to the mechanism designer.


A similar perspective is shared by Oury and Tercieux (2012), Chen, Mueller-Frank and Pai (2020) and Chen, Kunimoto and Sun (2020),\footnote{Oury and Tercieux (2012) write on page 1607 that ``\textit{These works (papers on robust prediction in games) share the common assumption that, in perturbed models, some types may have preferences over action profiles that are radically different from those of types in the initial model... Note that the meaning of such an assumption in the mechanism design context where the social planner fixes the game form is problematic}''.} who use an \textit{interim} approach to study robust partial implementation where agents' messages are assumed to be cheap talk.
These papers examine whether there exists a mechanism that partially implements a desired social choice function for \textit{all} nearby interim types.
By contrast, we take an \textit{ex ante} approach and examine whether the planner can robustly implement a desired social choice function with probability close to one when she knows that agents' beliefs are derived from a common prior and that the agents' payoff functions
coincide with those in her model with probability close to one.

\section{Main Results}\label{sec4}
Theorem \ref{Theorem1} shows that every $f$ is robustly implementable when agents' costs of learning are bounded from above.
Theorem \ref{Theorem2} shows that even when agents' costs of learning can be arbitrarily large, every $f$ is robustly implementable under a generic assumption on the state distribution. Theorem \ref{Theorem3} extends our robust implementation results
when agents' signals about the state are noisy.


\subsection{Robust Implementation under Bounded Cost}\label{sub4.1}
We show that every $f$ is robustly implementable when agents' costs of obtaining information is uniformly bounded from above.
\begin{Theorem}\label{Theorem1}
For every $\overline{c} \geq \max\{c_1,c_2\}$ and $f:\Theta \rightarrow \Delta (Y)$, there exists a finite mechanism $\mathcal{M}$ that robustly implements $f$ for all $\overline{c}$-bounded perturbations.
\end{Theorem}

For illustration purposes, we show this result (as well as Theorems \ref{Theorem2} and \ref{Theorem3}) in the leading example where
$u_i(\theta,y)=0$ for $i \in \{1,2\}$ and $c_1=c_2=c$, i.e., each normal type's payoff equals his transfer minus his cost of learning the state and the normal types of both agents face the same cost $c$. Extending our proof to general $(u_1,u_2)$ and heterogenous costs is straightforward.

\paragraph{Status Quo Rule:} Let $n \equiv |\Theta|$ and $\Theta \equiv \{\theta^1,...,\theta^n\}$. Each agent has $n$ messages $M_1=M_2 \equiv M \equiv \{1,2,...,n\}$. The outcome function is
\begin{equation}
g(m_1,m_2) = \left\{ \begin{array}{ll}
f(\theta^{m_1}) & \textrm{ if } m_1=m_2 \\
f(\theta^1) & \textrm{ otherwise. }
\end{array} \right.
\end{equation}
The transfer function for agent $i \in \{1,2\}$ is
\begin{equation}
t_i(m_i,m_{-i}) = \left\{ \begin{array}{ll}
R^j & \textrm{ if } m_1=m_2=j \\
0 & \textrm{ otherwise, }
\end{array} \right.
\end{equation}
where $R^n,...,R^1 >0$ and
$R^j \geq R^1+\frac{2\overline{c}}{q(\theta^j)}$ for every $j \geq 2$.

In the unperturbed game induced by our mechanism, an agent's pure strategy is an $n$-dimensional vector $(m^1,...,m^n)$ where $m^j \in M$ is the message he sends when the state is $\theta^j$. In order to capture agents' decisions to obtain information, each agent pays a penalty $c$ when he chooses a non-constant vector.
Let $\Sigma \equiv \{1,2,...,n\}^n$ be the set of pure strategies. An agent is \textit{truthful} if he uses strategy $(1,2,...,n)$, that is, he reports the index of the realized state. Our proof proceeds in three steps.

\paragraph{Step 1: Restricted Game without Perturbation} We start from examining a game \textit{without} any perturbation where agents are only allowed to use strategies in $\Delta (\Sigma^*)$, where $\Sigma^* \subset \Sigma$ is a subset of pure strategies defined as:
\begin{equation}
\Sigma^* \equiv \Big\{ (m^1,...,m^n) \in \Sigma \textrm{ such that } m^j \in \{1,j\} \textrm{ for every } j \geq 1 \Big\} .
\end{equation}
In this auxiliary game, each agent is only allowed to send the status quo message $1$ or truthfully report the state.
For example, when $n=2$, $\Sigma^* = \{(1,1),(1,2)\}$ while $\Sigma =\{(1,1),(1,2),(2,1),(2,2)\}$.
\begin{Lemma}\label{L1}
There exists $\gamma < \frac{1}{2}$ such that both agents being truthful is a $\gamma$-dominant equilibrium in the
restricted game without any perturbation.
\end{Lemma}
\begin{proof}
Conditional on $\theta=\theta^j$,
\begin{itemize}
  \item if agent 1 sends message $j$, his expected transfer equals $\Pr(m_2=j|\theta^j) R^j$;
  \item if agent 1 sends message $1$, his expected transfer equals $\Pr(m_2=1|\theta^j) R^1$.
\end{itemize}
If agent 2's strategy belongs to $\Delta (\Sigma^*)$ and agent 2 is truthful with probability at least $\frac{1}{2}$, we have $\Pr(m_2=j|\theta^j) \geq \frac{1}{2}$, so $\Pr(m_2=j|\theta^j) R^j \geq \Pr(m_2=1|\theta^j) R^1$ given the condition that $R^j>R^1$. Since $R^j>R^1+\frac{2\overline{c}}{q(\theta^j)}$, agent 1 strictly prefers to send message $j$ to any $m_1 \leq 1$ in state $\theta^j$
as long as his cost of observing $\theta$ is no more than $\overline{c}$.
Since agent 1's incentives are strict when he believes that agent 2 is truthful with probability at least $\frac{1}{2}$, there exists $\gamma < \frac{1}{2}$ such that agent 1 strictly prefers to use strategy $(1,2,...,n)$ to other strategies in $\Sigma^*$ when agent $2$'s strategy belongs to $\Delta (\Sigma^*)$ and is truthful with probability at least $\gamma$.
\end{proof}

\paragraph{Step 2: Restricted Game with Perturbation} For any perturbation $\mathcal{G}$, consider a \textit{restricted perturbed game} where type $Q_i(\omega)$ of agent $i$'s payoff is $\widetilde{u}_i(\omega,\theta,y)+t_i-c \chi_i$ and all types of both agents are only allowed to use strategies in $\Delta(\Sigma^*)$. Since both agents being truthful is a $\gamma$-dominant equilibrium in the unperturbed restricted game for some $\gamma< \frac{1}{2}$, the critical path lemma in Kajii and Morris (1997) implies the following result.
\begin{Lemma}\label{L2}
For every $\varepsilon>0$, there exists $\eta>0$, such that for every $\eta$-perturbation $\mathcal{G}$,
there exists an equilibrium
$\sigma(\mathcal{G})$ when the environment is perturbed by $\mathcal{G}$ and all types of both agents are only allowed to use strategies in $\Delta (\Sigma^*)$ such that
the probability with which both agents are truthful in this equilibrium is at least $1-\varepsilon$.
\end{Lemma}
Since $g(j,j)=f(\theta^j)$ for every $j \in \{1,2,...,n\}$,
$f$ is implemented with probability more than $1-\varepsilon$ if both agents behave according to $\sigma(\mathcal{G})$. What remains to be verified is that $\sigma(\mathcal{G})$ remains an equilibrium when agents can use any strategy in $\Delta (\Sigma)$.

\paragraph{Step 3: Unrestricted Game with Perturbation} We show that for every $\mathcal{G}$, $\sigma(\mathcal{G})$ remains an equilibrium under perturbation $\mathcal{G}$ when agents can use any strategy in $\Delta (\Sigma)$. Suppose by way of contradiction that there exists a type $Q_1(\omega)$ of agent $1$ who strictly prefers strategy $(m^1,...,m^n) \notin \Sigma^*$ to any strategy in $\Sigma^*$ when agent 2 behaves according to $\sigma(\mathcal{G})$. Define a new strategy $(m_*^1,...,m_*^n)$ such that
  \begin{equation*}
m_{*}^j \equiv \left\{ \begin{array}{ll}
m^j & \textrm{ if } m^j \in \{1,j\}\\
1 & \textrm{ if } m^j \notin \{1,j\}
\end{array} \right. \textrm{ for every } j \in \{1,2,...,n\}.
\end{equation*}
By construction, $(m_*^1,...,m_*^n) \in \Sigma^*$. We compare type $Q_1(\omega)$'s expected payoff from $(m^1,...,m^n)$ and from $(m_*^1,...,m_*^n)$. First, $(m^1,...,m^n)$ and $(m_*^1,...,m_*^n)$ lead to the same joint distribution of $(\theta,y)$ given that agent 2's strategy belongs to $\Delta (\Sigma^*)$, since conditional on $\theta=\theta^j$, agent $2$ either sends either $1$ or $j$, so the implemented outcome is $f(\theta^1)$ when agent $1$'s message is neither $1$ nor $j$. Second, $(m_*^1,...,m_*^n)$ leads to weakly greater transfers conditional on each state since the transfer is $0$ when agent $1$ sends message $m^j \notin \{1,j\}$ in state $\theta^j$ given that agent 2's message belongs to $\{1,j\}$. Third, $(m_*^1,...,m_*^n)$ leads to strictly greater transfers compared to $(m^1,...,m^n)$ when $(m_*^1,...,m_*^n)$ requires strictly greater learning cost. To see this, note that $(m_*^1,...,m_*^n)$ requires a greater cost only if $m^1=...=m^n$ and, since $(m^1,...,m^n) \notin \Sigma^*$, it must be that $m^1=...=m^n  \geq 2$. As a result, $(m_*^1,...,m_*^n)$ leads to strictly greater transfer conditional on state $\theta^1$, and the expected increase in transfer is at least $R^1 q(\theta^1)$, which is strictly greater than the maximal cost of learning $\overline{c}$. This suggests that every type of agent 1 prefers $(m_*^1,...,m_*^n)$ to $(m^1,...,m^n)$, which leads to a contradiction. Hence, $\sigma(\mathcal{G})$ remains an equilibrium when agents can use any strategy in $\Delta (\Sigma)$.

\subsection{Robust Implementation under Generic State Distribution}\label{sub4.2}
We show that as long as the objective state distribution $q$ satisfies a generic condition, every $f$ is robustly implementable even when agents' costs can be unbounded in perturbed environments.
\begin{Definition}
The objective state distribution $q \in \Delta (\Theta)$ is generic if there exists $\theta^* \in \Theta$ such that $q(\theta^*) > q(\theta')$ for every $\theta' \neq \theta^*$.
\end{Definition}
For example, when there are two states, our generic condition rules out $q$ that assigns probability $\frac{1}{2}$ to each state but allows for any other full support distribution.
\begin{Theorem}\label{Theorem2}
Suppose $q$ is generic. For every social choice function $f:\Theta \rightarrow \Delta (Y)$, there exists a finite mechanism $\mathcal{M}$ that robustly implements $f$.
\end{Theorem}

\paragraph{Augmented Status Quo Rule:} Recall that $\Theta \equiv \{\theta^1,...,\theta^n\}$.
When $q$ is generic, there exists a ranking over states such that $q(\theta^1)> q(\theta^2) \geq q(\theta^3) \geq ... \geq q(\theta^n)>0$.

Each agent has $2n-1$ messages:
$M_1=M_2= \{-n,...,-2\} \cup \{1\} \cup \{2,...,n\}$.
The outcome function is
\begin{equation}\label{3.1}
g(m_1,m_2) = \left\{ \begin{array}{ll}
f(\theta^{|m_1|}) & \textrm{ if } |m_1|=|m_2| \\
f(\theta^1) & \textrm{ otherwise. }
\end{array} \right.
\end{equation}
The transfer function for agent $i \in \{1,2\}$ is
\begin{equation}\label{3.2}
t_i(m_i,m_{-i}) = \left\{ \begin{array}{ll}
R^j & \textrm{ if } m_1=m_2=j \geq 1 \\
R^0 & \textrm{ if } m_1,m_2 \leq 1 \textrm{ but } (m_1,m_2) \neq (1,1)\\
0 & \textrm{ otherwise, }
\end{array} \right.
\end{equation}
where $R^n>R^{n-1}>...>R^2>R^1>R^0>0$ satisfy
\begin{equation}\label{3.3}
R^1>R^0+\frac{2c}{q(\theta^1)} \quad \textrm{and} \quad   R^j \geq R^1+\frac{2c}{q(\theta^j)} \textrm{ for every } j  \geq 2
\end{equation}
and
\begin{equation}\label{3.4}
\frac{R^0}{R^n} > \frac{q(\theta^2)}{q(\theta^1)}.
\end{equation}
When $q$ is generic, there exist $R^n,...,R^1,R^0$ that satisfy both (\ref{3.3}) and (\ref{3.4}). The case with two states has been shown in the example. When there are three states, our mechanism is given by:
\begin{center}
\begin{tabular}{| c | c | c | c | c | c |}
\hline
  $g$ & $-3$ & $-2$ & 1 & 2 & 3 \\
  \hline
  $-3$ & ${\color{blue}{f(\theta^3)}}$ & $f(\theta^1)$ & $f(\theta^1)$ & $f(\theta^1)$ & ${\color{blue}{f(\theta^3)}}$\\
  \hline
  $-2$ & $f(\theta^1)$ & ${\color{blue}{f(\theta^2)}}$ & $f(\theta^1)$ & ${\color{blue}{f(\theta^2)}}$ & $f(\theta^1)$\\
  \hline
  1 & $f(\theta^1)$ & $f(\theta^1)$ & $f(\theta^1)$ & $f(\theta^1)$   & $f(\theta^1)$\\
  \hline
  2 & $f(\theta^1)$ & ${\color{blue}{f(\theta^2)}}$ & $f(\theta^1)$ & ${\color{blue}{f(\theta^2)}}$ & $f(\theta^1)$\\
  \hline
  3 & ${\color{blue}{f(\theta^3)}}$ & $f(\theta^1)$ & $f(\theta^1)$ & $f(\theta^1)$ & ${\color{blue}{f(\theta^3)}}$\\
  \hline
\end{tabular}
\end{center}
\begin{center}
\begin{tabular}{| c | c | c | c | c | c |}
\hline
  $t_1,t_2$ & $-3$ & $-2$ & 1 & 2 & 3 \\
  \hline
  $-3$ & $R^0, R^0$ & $R^0, R^0$ & $R^0, R^0$ & $0,0$ & $0,0$\\
  \hline
  $-2$ & $R^0, R^0$ & $R^0, R^0$ & $R^0, R^0$ & $0,0$ & $0,0$\\
  \hline
  1 & $R^0, R^0$ & $R^0, R^0$ & ${\color{blue}{R^1,R^1}}$ & $0,0$   & $0,0$\\
  \hline
  2 & $0,0$ & $0,0$ & $0,0$ & ${\color{blue}{R^2,R^2}}$ & $0,0$\\
  \hline
  3 & $0,0$ & $0,0$ & $0,0$ & $0,0$ & ${\color{blue}{R^3,R^3}}$\\
  \hline
\end{tabular}
\end{center}

An agent's pure strategy in the unperturbed environment is $(m^1,...,m^n)$ where $m^j \in M$ represents the message he sends when the state is $\theta^j$. Let $\Sigma \equiv \{-n,...,-2,1,2,...,n\}^n$ be the set of pure strategies. An agent is \textit{truthful} if he uses strategy $(1,2,...,n)$, that is, he reports the index of the realized state. Our proof proceeds in three steps, which is similar to that of Theorem \ref{Theorem1}.

\paragraph{Step 1: Restricted Game without Perturbation} We start from examining a game \textit{without} any perturbation where agents are only allowed to use strategies in $\Delta (\Sigma^*)$, where
\begin{equation}\label{3.5}
\Sigma^* \equiv \Big\{ (m^1,...,m^n) \in \Sigma \textrm{ such that } m^j \in \{-n,...,-2,1\}\cup\{j\} \textrm{ for every } j \geq 1 \Big\} .
\end{equation}
In words, each agent is only allowed to send negative messages, the status quo message $1$, or a message that coincides with the index of the realized state. For example, when $n=2$, $\Sigma^* = \{(-2,-2),(-2,1),(-2,2),(1,-2), (1,1),(1,2)\}$ while $\Sigma =\Sigma^* \bigcup \{(2,-2),(2,1),(2,2)\}$.

We show that there exists $\gamma < \frac{1}{2}$ such that both agents being truthful is a $\gamma$-dominant equilibrium in the
restricted game without any perturbation. Suppose agent $1$ believes that agent $2$ plays $(1,2,...,n)$ with probability at least $\frac{1}{2}$ and that agent 2's strategy belongs to $\Delta (\Sigma^*)$.
\begin{itemize}
  \item Conditional on $\theta=\theta^j$ for every $j \in \{2,3,...,n\}$. If agent 1 sends message $j$,
  his expected transfer equals $\Pr(m_2=j|\theta^j) R^j$, which is at least $\frac{R^j}{2}$ given that agent 2 is truthful with probability at least $\frac{1}{2}$. If agent $1$ sends any $m_1 \leq 1$, his expected transfer  is no more than $\Pr(m_2 \neq j |\theta^j) R^1$, which is at most $\frac{R^1}{2}$. Since $R^j>R^1+\frac{2c}{q(\theta^j)}$, agent 1 strictly prefers message $j$ to any $m_1 \leq 1$ in state $\theta^j$ even taking into account the cost $c$ of observing $\theta$.
  \item Conditional on $\theta=\theta^1$. If agent 1 sends message $1$, his expected transfer is $\Pr(m_2=1|\theta^1) R^1 + \Pr(m_2 \leq -2 |\theta^1) R^0$, which is at least $\frac{R^1+R^0}{2}$ given that $\Pr(m_2=1|\theta^1) \geq \frac{1}{2}$. If agent 1 sends any negative message, he receives transfer $R^0$. Since $R^1>R^0+\frac{2c}{q(\theta^1)}$, agent 1 strictly prefers $1$ to any negative message in state $\theta^1$ even taking into account the cost of observing $\theta$.
\end{itemize}
Since agent 1's incentives are strict when he believes that agent 2 is truthful with probability at least $\frac{1}{2}$, there exists $\gamma < \frac{1}{2}$ such that agent 1 strictly prefers $(1,2,...,n)$ to other strategies in $\Sigma^*$ when agent $2$'s strategy belongs to $\Delta (\Sigma^*)$ and is truthful with probability at least $\gamma$.

\paragraph{Step 2: Restricted Game with Perturbation} For any perturbation $\mathcal{G}$, consider a \textit{restricted perturbed game} where type $Q_i(\omega)$ of agent $i$'s payoff is $\widetilde{u}_i(\omega,\theta,y)+t_i-\widetilde{c}_i(\omega) \chi_i$ and all types of both agents are only allowed to use strategies in $\Delta(\Sigma^*)$.

Since both agents being truthful is a $\gamma$-dominant equilibrium in the unperturbed restricted game for some $\gamma$ less than $\frac{1}{2}$,
the critical path lemma in Kajii and Morris (1997) implies that for every $\varepsilon>0$, there exists $\eta>0$, such that for every $\eta$-perturbation $\mathcal{G}$,
there exists an equilibrium
$\sigma(\mathcal{G})$ in the perturbed game where all types of both agents are only allowed to use strategies in $\Delta (\Sigma^*)$ such that
the probability with which both agents are truthful under $\sigma(\mathcal{G})$ is at least $1-\varepsilon$.

\paragraph{Step 3: Unrestricted Game with Perturbation} We show that when $q$ is generic and $\{R^n,...,R^1,R^0\}$ satisfy (\ref{3.3}) and (\ref{3.4}), for every perturbation $\mathcal{G}$, every equilibrium $\sigma(\mathcal{G})$ in the restricted perturbed game remains an equilibrium in the unrestricted perturbed game when both agents can use any strategy in $\Delta(\Sigma)$. For this purpose, we only need to show that for every pure strategy that does not belong to $\Sigma^*$, there exists a pure strategy that belongs to $\Sigma^*$ such that every type of agent $1$ weakly prefers
the latter to the former when he believes that
agent 2 plays according to $\sigma(\mathcal{G})$. We consider two cases separately.

For every $(m^1,...,m^n) \notin \Sigma^*$ that is non-constant, let $(m^1_*,...,m^n_*)$ be such that
  \begin{equation}
m_{*}^j \equiv \left\{ \begin{array}{ll}
m^j & \textrm{ if } m^j \in \{-n,...,-2\} \cup \{1,j\} \\
-m^j & \textrm{ if } m^j \notin \{-n,...,-2\} \cup \{1,j\}
\end{array} \right. \textrm{ for every } j \in \{1,2,...,n\}.
\end{equation}
By construction, $(m^1_*,...,m^n_*) \in \Sigma^*$, and moreover, it does not increase the cost of obtaining information compared to $(m^1,...,m^n)$. Our construction of $g (m_1,m_2)$ ensures that $(m^1_*,...,m^n_*)$ and $(m^1,...,m^n)$ induce the same distribution of $(\theta,y)$ regardless of agent 2's message. Regardless of agent 1's type, as long as he believes that agent 2 behaves according to $\sigma(\mathcal{G})$, which implies that agent 2's strategy belongs to $\Delta (\Sigma^*)$, agent 1 receives weakly greater transfer from $(m^1_*,...,m^n_*)$ compared to $(m^1,...,m^n)$ since sending $m^j \notin \{-n,...,-2\} \cup \{1,j\}$ leads to a transfer of $0$ in state $\theta^j$ when agent 2's message belongs to $\{-n,...,-2\} \cup \{1,j\}$.

For every $(m^1,...,m^n) \notin \Sigma^*$ that is constant, there exists $k \in \{2,3,...,n\}$ such that $(m^1,...,m^n)=(k,...,k)$. Compare any given type of agent 1's payoff from using strategy $(k,...,k)$ and from using strategy $(-k,...,-k)$. Our construction of $g(m_1,m_2)$ implies that $(k,...,k)$ and  $(-k,...,-k)$
lead to the same distribution over $(\theta,y)$ and neither strategy requires the agent to learn $\theta$. The expected transfer agent 1 receives is $\Pr(m_2=k)R^k$ if he uses strategy $(k,...,k)$ and is $\Pr(m_2 \leq 1) R^0$ if he uses strategy $(-k,...,-k)$. When every type of agent 2's strategy belongs to $\Delta (\Sigma^*)$, $\Pr(m_2 \leq 1) \geq q(\theta^1)$ and $\Pr(m_2=k) \leq q(\theta^k)$. We have $\Pr(m_2=k)R^k < \Pr(m_2 \leq 1) R^0$ given (\ref{3.4}).

\subsection{Robustness to Trembles and Noisy Information}\label{sub4.3}
We extend our robust implementation result to situations where agents may tremble with small probability when sending messages and agents' signals about the state can be noisy (so that the two agents' private signals may not be perfectly correlated).
Our new mechanism has the same outcome function as the augmented status quo rule but with a different transfer function.

\paragraph{Trembles:} For any mechanism $\mathcal{M}$, suppose for every $i \in \{1,2\}$,
  when agent $i$ intends to send message $m_i \in M_i$, the principal receives $m_i$ with probability $1-\tau$ and receives a message that is drawn according to $F_i^{M_i} \in \Delta (M_i)$ with probability $\tau$.

  Throughout this section, we distinguish between an agent's \textit{intended message} and his \textit{realized message}. We write $F_i$ instead of $F_i^{M_i}$ in order to simplify notation.

\paragraph{Noisy Information:} Suppose $q \in \Delta (\Theta)$ is generic. Let $\Theta \equiv \{\theta^1,...,\theta^n\}$  such that $q(\theta^1)> q(\theta^2) \geq ... \geq q(\theta^n)>0$. For every $i \in \{1,2\}$, let $S_i \equiv \{s_i^1,...,s_i^{|S_i|}\}$ be agent $i$'s signal space, with $|S_i| \geq n$. Note that $|S_i|$ can be any finite number, i.e., there is no known upper bound on the number of signal realizations.
Let $\pi \in \Delta (\Theta \times S_1 \times S_2)$ be the joint distribution of the state and agents' private signals. We say that $\pi$ is of size $\tau>0$ if
\begin{itemize}
  \item[(a)] The marginal distribution of $\pi$ on $\Theta$ is $q \in \Delta (\Theta)$.
  \item[(b)] For every $i \in \{1,2\}$, there exists a mapping $h_i : S_i \rightarrow \{1,2,...,n\}$  such that
  \begin{equation}\label{A.1}
    \pi \Big( h_{-i}(s_{-i})=h_i(s_i) \Big| s_i \Big) \geq 1-\tau \textrm{ for every } s_i \in S_i,
  \end{equation}
and
  \begin{equation}\label{A.2}
  \sum_{j=1}^n \sum_{s_i \in \{h_i(s_i)=j\}} \pi(\theta^j,s_i) \geq 1-\tau.
  \end{equation}
\end{itemize}
Our first requirement says that the marginal distribution on $\theta$ equals $q$. Our second requirement is reminiscent of Chung and Ely (2003) and Sugaya and Takahashi (2013),
in which every signal realization is linked to a particular state, given by the mapping $h_i$.
One can think about $h_i$ as endowing each signal realization with a \textit{meaning}.
The mappings from signal realizations to their meanings satisfy, first,
every agent believes that the other agent receives a signal with the same meaning with probability close to $1$, and second, the meaning of each agent's signal realization coincides with the state with ex ante probability close to $1$.

The planner does not know the perturbation $\mathcal{G}$ as well as $\{\tau, F_1, F_2,\pi\}$. She would like to design a mechanism that can approximately implement the desired social choice function for all small enough perturbations, small enough trembles, and small enough noise in agents' signals about the state. Agents know the mechanism $\mathcal{M}$, the perturbation $\mathcal{G}$, their respective information about $\omega$ under $\mathcal{G}$, as well as $\{\tau, F_1, F_2,\pi\}$ before deciding whether to learn $\theta$ and which messages they intend to send. The planner implements an outcome and makes transfers according to the \textit{realized messages} and the mechanism she committed to.
\begin{Theorem}\label{Theorem3}
Suppose $q$ is generic. For every $f: \Theta \rightarrow \Delta (Y)$, there exists $\mathcal{M}$ such that for every $\varepsilon>0$, there exist $\eta>0$ and $\overline{\tau}>0$ such that for every $\tau< \overline{\tau}$, $F_1,F_2$, every $\pi$ that is of size $\overline{\tau}$, and every $\eta$-perturbation $\mathcal{G}$, there exists an equilibrium $\sigma(\mathcal{G})$ such that $\max_{\theta \in \Theta} ||    g_{\sigma(\mathcal{G})}(\theta)-f(\theta) ||_{\textrm{TV}} < \varepsilon$.
\end{Theorem}
Similar to our proof of Theorem \ref{Theorem2}, we consider a mechanism where each agent has $2n-1$ messages $M \equiv \{-n,...,-2\} \cup \{1\} \cup \{2,3,...,n\}$.
Agent $i$'s pure strategy is an $|S_i|$-dimensional vector $(m^1,...,m^{|S_i|})$ where $m^k \in M$ represents
agent $i$'s \textit{intended message}
when his private signal is $s_i=s_i^k$. That being said, conditional on
$s_i=s_i^k$,
 agent $i$'s \textit{realized message}  is $m^k$ with probability $1-\tau$ and is randomly drawn according to $F_i \in \Delta (M_i)$ with probability $\tau$.
Let
\begin{equation*}
\Sigma_i^* \equiv \Big\{ (m^1,...,m^{|S_i|}) \in \Sigma \textrm{ such that for  every $k$, } m^k \in \{-n,...,-2,1\}\cup\{j\} \textrm{ when } h_i(m^k)=j \Big\}.
\end{equation*}
Intuitively, $\Sigma^*$ is the subset of agent $i$'s pure strategies such that conditional on every realization of his private signal, he either sends a negative message, or the status quo message $1$, or the true meaning of his private signal.
Agent $i$ \textit{intends to be truthful} if he sends the true meaning of his private signal for every $s_i \in S_i$.

When there are two states, our augmented status quo rule can robustly implement $f$ when agents' tremble and their signals about the state are noisy.
When there are three or more states, our augmented status quo rule cannot robustly implement $f$ as can be illustrated by the following example.
Suppose there are three states $\{\theta^1,\theta^2,\theta^3\}$ and three private signals $\{s_i^1,s_i^2,s_i^3\}$ for each agent $i \in \{1,2\}$, where $h_i(s_i^j)=j$ for every $i \in \{1,2\}$ and $j \in \{1,2,3\}$. For simplicity, let us also assume that each agent's private signal perfectly reveals the state. Suppose agent $1$ observes $s_1=s_1^3$ and he believes that every type of agent 2's strategy belongs to $\Delta (\Sigma_2^*)$, let us compare agent 1's expected transfer when he intends to send message $-2$ and when he intends to send message $2$.
When agent 1's realized message is $2$, his expected transfer under our augmented status quo rule is $\Pr(m_2=2|\theta=\theta^3) R^{2}$.
When agent 1's realized message is $-2$, his expected transfer under our augmented status quo rule is $\Pr(m_2 \leq 1|\theta=\theta^3) R^{0}$.
If the trembling probability $\tau$ is $0$, then $\Pr(m_2 \leq 1|\theta=\theta^3) R^{0} \geq \Pr(m_2=2|\theta=\theta^3) R^{2}$ when agent 2's strategy belongs to $\Delta(\Sigma^*)$.
If $\tau>0$ and agent 2 intends to send message $3$ with probability $1$ in state $\theta^3$, then $\Pr(m_2=2|\theta=\theta^3) R^{2}=\tau F_2(2) R^{2}$ can be strictly greater than
$\Pr(m_2 \leq 1| \theta=\theta^3) R^0=\tau F_2(m_2 \leq 1) R^0$.

We present a new mechanism called the \textit{Modified Status Quo Rule} which solves this problem. Each agent has $2n-1$ messages $M \equiv \{-n,...,-2\} \cup \{1\} \cup \{2,...,n\}$. The outcome function is
\begin{displaymath}
g(m_1,m_2) = \left\{ \begin{array}{ll}
f(\theta^{|m_1|}) & \textrm{ if } |m_1|=|m_2| \\
f(\theta^1) & \textrm{ otherwise }
\end{array} \right.
\end{displaymath}
The transfer functions are
\begin{displaymath}
t_1(m_1,m_2) = \left\{ \begin{array}{ll}
R^j & \textrm{ if } m_1=m_2=j \geq 1 \\
R^0 & \textrm{ if } m_1 \leq 1 \textrm{ but } (m_1,m_2) \neq (1,1)\\
R^0-x & \textrm{ if } m_1 \geq 2 \textrm{ and } m_2 \leq 1 \\
0 & \textrm{ otherwise }
\end{array} \right.
\end{displaymath}
\begin{displaymath}
t_2(m_1,m_2) = \left\{ \begin{array}{ll}
R^j & \textrm{ if } m_1=m_2=j \geq 1 \\
R^0 & \textrm{ if } m_2 \leq 1 \textrm{ but } (m_1,m_2) \neq (1,1)\\
R^0-x & \textrm{ if } m_2 \geq 2 \textrm{ and } m_1 \leq 1 \\
0 & \textrm{ otherwise }
\end{array} \right.
\end{displaymath}
where $R^n,...,R^0>x>\frac{c}{q(\theta^n)}$ satisfy
\begin{equation}\label{3.7}
  R^1-R^0 >  \frac{4c}{q(\theta^1)}, \quad  R^{j}-R^{1}-x > \frac{2c}{q(\theta^j)} \textrm{ for every } j \in \{2,3,...,n\},
\end{equation}
and
\begin{equation}\label{3.8}
   \frac{x}{R^j-R^0} > \frac{q(\theta^j)}{q(\theta^1)} \textrm{ for every } j \in \{2,3,...,n\}.
\end{equation}
When there are three states, the modified status quo rule is given by:
\begin{center}
\begin{tabular}{| c | c | c | c | c | c |}
\hline
  $g$ & $-3$ & $-2$ & 1 & 2 & 3 \\
  \hline
  $-3$ & ${\color{blue}{f(\theta^3)}}$ & $f(\theta^1)$ & $f(\theta^1)$ & $f(\theta^1)$ & ${\color{blue}{f(\theta^3)}}$\\
  \hline
  $-2$ & $f(\theta^1)$ & ${\color{blue}{f(\theta^2)}}$ & $f(\theta^1)$ & ${\color{blue}{f(\theta^2)}}$ & $f(\theta^1)$\\
  \hline
  1 & $f(\theta^1)$ & $f(\theta^1)$ & $f(\theta^1)$ & $f(\theta^1)$   & $f(\theta^1)$\\
  \hline
  2 & $f(\theta^1)$ & ${\color{blue}{f(\theta^2)}}$ & $f(\theta^1)$ & ${\color{blue}{f(\theta^2)}}$ & $f(\theta^1)$\\
  \hline
  3 & ${\color{blue}{f(\theta^3)}}$ & $f(\theta^1)$ & $f(\theta^1)$ & $f(\theta^1)$ & ${\color{blue}{f(\theta^3)}}$\\
  \hline
\end{tabular}
\end{center}
\begin{center}
\begin{tabular}{| c | c | c | c | c | c |}
\hline
  $t_1,t_2$ & $-3$ & $-2$ & 1 & 2 & 3 \\
  \hline
  $-3$ & $R^0, R^0$ & $R^0, R^0$ & $R^0, R^0$ & $R^0,R^0{\color{red}{-x}}$ & $R^0,R^0{\color{red}{-x}}$\\
  \hline
  $-2$ & $R^0, R^0$ & $R^0, R^0$ & $R^0, R^0$ & $R^0,R^0{\color{red}{-x}}$ & $R^0,R^0{\color{red}{-x}}$\\
  \hline
  1 & $R^0, R^0$ & $R^0, R^0$ & ${\color{blue}{R^1,R^1}}$ & $R^0,R^0{\color{red}{-x}}$   & $R^0,R^0{\color{red}{-x}}$\\
  \hline
  2 & $R^0{\color{red}{-x}},R^0$ & $R^0{\color{red}{-x}},R^0$ & $R^0{\color{red}{-x}},R^0$ & ${\color{blue}{R^2,R^2}}$ & $0,0$\\
  \hline
  3 & $R^0{\color{red}{-x}},R^0$ & $R^0{\color{red}{-x}},R^0$ & $R^0{\color{red}{-x}},R^0$ & $0,0$ & ${\color{blue}{R^3,R^3}}$\\
  \hline
\end{tabular}
\end{center}

Intuitively, the outcomes under the augmented status quo rule and the modified status quo rule are the same. The only difference is in the transfer function: By sending the status quo message or any negative message, an agent is guaranteed to receive transfer $R^0$. When an agent sends a message at least $2$, he faces a penalty $x$ if the other agent sends the status quo message or a negative message, and receives zero transfer when the other agent sends a different message of at least $2$.

The proof follows similar steps as before. First, there exists $\gamma < \frac{1}{2}$ such that both agents intending to be truthful is a $\gamma$-dominant equilibrium in the restricted unperturbed game where agents are only allowed to use strategies in $\Delta(\Sigma_1)$ and $\Delta (\Sigma_2)$. To see this, suppose agent 2 intends to be truthful with probability at least $\frac{1}{2}$.
\begin{itemize}
  \item For every $j \geq 2$,
conditional on every $s_1 \in S_1$ with $h(s_1)=j$, if agent $1$'s realized message is $j$, then he receives an expected transfer of
\begin{equation*}
    \Pr(m_2=j|s_1) R^j +\Pr(m_2 \leq 1 |s_1) R^0,
\end{equation*}
and if agent 1's realized message is no more than 1, then he receives an expected transfer of
\begin{equation*}
    \Pr(m_2 =1 |s_1) R^1 + \Pr(m_2 \neq 1 | s_1) R^0.
\end{equation*}
Since $\pi(h_{2}(s_{2})=h_1(s_1)|s_1) \geq 1-\overline{\tau}$ when $\pi$ is of size $\overline{\tau}$, we have
$\Pr(m_2=j|s_1) \geq \frac{1-\tau}{2} (1-\overline{\tau})$ and $\Pr(m_2=1|s_1) \leq 1- \frac{1-\tau}{2} (1-\overline{\tau})$.
When $R^{j}-R^{1}-x > \frac{2c}{q(\theta^j)}$, $\overline{\tau}$ is small enough, and $\tau \leq \overline{\tau}$, we have
\begin{equation*}
    \Pr(m_2=j|s_1) R^j +\Pr(m_2 \leq 1 |s_1) R^0>\Pr(m_2 =1 |s_1) R^1 + \Pr(m_2 \neq 1 | s_1) R^0.
\end{equation*}
Hence, agent 1 strictly prefers to send message $j$ when he receives signal $s_1$ such that $h(s_1)=j$.
  \item When $R^1-R^0 >  \frac{4c}{q(\theta^1)}$,  conditional on agent 2's strategy belongs to $\Delta (\Sigma_2^*)$ and agent 2 intends to be truthful with probability at least $\frac{1}{2}$,
agent 1 intending to send message $1$ when $h(s_1)=1$ leads to a strictly greater transfer compared to him intending to send any negative message.
\end{itemize}

The second step uses the critical path lemma. We can show that for every $\varepsilon>0$, there exists $\eta>0$ such that for every $\eta$-perturbation $\mathcal{G}$, there exists an equilibrium $\sigma(\mathcal{G})$ in the perturbed restricted game where both agents intend to be truthful with probability more than $1-\frac{\varepsilon}{2}$. Under the outcome function $g$ of our mechanism, if both agents behave according to  $\sigma(\mathcal{G})$ and $\overline{\tau}$ is small compared to $\varepsilon$, then the implemented outcome is $\varepsilon$-close to $f(\theta)$ conditional on every $\theta$.

For the third step, for every strategy of agent $1$'s $(m^1,...,m^{|S_1|}) \notin \Sigma_1^*$ that is non-constant,
let $(m^1_*,...,m^{|S_1|}_*) \in \Sigma^*$ be defined as:
  \begin{equation*}
m_{*}^k \equiv \left\{ \begin{array}{ll}
m^k & \textrm{ if } m^k \in \{-n,...,-2\} \cup \{1,j\} \textrm{ and } h_1(s_1^k)=j\\
-m^k & \textrm{ if } m^k \notin \{-n,...,-2\} \cup \{1,j\} \textrm{ and } h_1(s_1^k)=j
\end{array} \right. \textrm{ for every } k \in \{1,2,...,|S_1|\}.
\end{equation*}
Intuitively, for every signal realization $s_1^k$, $m_*^k=m^k$ if $m^k$ is no more than $1$ or $m^k$ coincides with the meaning of $s_1^k$; otherwise, $m_*^k=-m^k$. By construction, $(m^1,...,m^{|S_1|})$ and $(m^1_*,...,m^{|S_1|}_*)$ induce the same joint distribution of $(\theta,y)$.

We compare agent 1's expected transfer from $(m^1,...,m^{|S_1|})$ and from $(m^1_*,...,m^{|S_1|}_*)$.
When agent $1$'s private signal $s_1$ is such that $h(s_1)=j$, agent 1's expected transfer when his realized message $m \notin \{-n,...,-2\} \cup \{1,j\}$ is:
\begin{equation}\label{3.9}
    \Pr(m_2=m|s_1) R^{m} + \Pr(m_2 \leq 1|s_1) (R^0-x).
\end{equation}
His expected transfer when his realized message is $-m$ is $R^0$.
Since $\Pr(m_2=m|s_1) \leq 2 \overline{\tau}$ when agent 2's strategy belongs to $\Delta (\Sigma^*)$, the value of (\ref{3.9}) is strictly less than $R^0$ when $\overline{\tau}$ is small. This implies that every type of agent 1 prefers
$(m^1_*,...,m^{|S_1|}_*)$ to $(m^1,...,m^{|S_1|})$.

For every $(m^1,...,m^{|S_1|}) \notin \Sigma_1^*$ that is constant, there exists $k \in \{2,3,...,n\}$ such that $(m^1,...,m^{|S_1|})=(k,...,k)$.
Compare agent 1's expected transfer (unconditioned on $\theta$, $s_1$, and $s_2$) when his realized message is $k$ and when his realized message is $-k$. When his realized message is $k$, he receives an expected transfer of $\Pr(m_2=k) R^k +\Pr(m_2 \leq 1) (R^0-x)$. When his realized message is $-k$, he receives an expected transfer of $R^0$.
When agent 2's strategy belongs to $\Delta (\Sigma_2^*)$,
\begin{equation*}
    \Pr(m_2=k) \leq \pi(h_2(s_2)=k) + \Big( 1-\pi(h_2(s_2)=k) \Big) \overline{\tau} \quad \textrm{and} \quad
\Pr(m_2 \leq 1) \geq \pi(h_2(s_2)=1) (1-\overline{\tau}).
\end{equation*}
Hence, $\Pr(m_2=k) R^k +\Pr(m_2 \leq 1) (R^0-x)< R^0$
when $\overline{\tau}$ is small enough. Hence, conditional on agent 2 behaves according to $\sigma(\mathcal{G})$, every type of agent 1 prefers strategy $(-k,...,-k)$ to $(k,k,...k)$ for every $k \geq 2$.

\section{Stronger Notions of Robust Implementation}\label{sec5}
First, we show that the planner \textit{cannot} implement any state-contingent social choice function when we allow for perturbations where agents' payoffs do not coincide with those in the unperturbed environment with high probability.
Second, we show that the planner cannot fully or virtually implement any state-contingent social choice function when agents' payoff functions do not depend on the state, or when agents' costs of learning are above some cutoff. We also provide a sufficient condition on agents' payoff functions under which full implementation is plausible when the costs of learning are sufficiently small. Throughout this section, we focus on non-constant $f$:
\begin{Definition}
$f$ is non-constant if there exist $\theta,\theta'$ such that $f(\theta) \neq f(\theta')$.
\end{Definition}
\subsection{Impossibility of Global Robust Implementation}\label{sub4.4}
We show that for every $\overline{c} \geq \max \{\overline{c}_1,\overline{c}_2\}$, no finite mechanism can approximately implement any non-constant social choice function for all $\overline{c}$-bounded perturbations.
\begin{Proposition}\label{Prop1}
For every $\overline{c}\geq \max \{\overline{c}_1,\overline{c}_2\}$ and every $f: \Theta \rightarrow \Delta (Y)$ that is non-constant, there exists no finite mechanism that can globally implement $f$ for all $\overline{c}$-bounded perturbations.
\end{Proposition}
Proposition \ref{Prop1} implies that no finite mechanism can approximately implement any state-contingent $f$ for all perturbations. It also implies the following corollary:
\begin{Corollary}
For every $f: \Theta \rightarrow \Delta (Y)$ that is non-constant, there exists $k(f)>0$ such that for every finite mechanism $\mathcal{M}$ and every $\eta>0$, there exists a $\overline{c}$-bounded $\eta$-perturbation $\mathcal{G}$, such that for every equilibrium $\sigma (\mathcal{G})$ of the game $(\mathcal{M},\mathcal{G})$, we have $\max_{\theta \in \Theta} ||    g_{\sigma(\mathcal{G})}(\theta)-f(\theta) ||_{\textrm{TV}} \geq \eta k(f)$.
\end{Corollary}

The above corollary implies that when we allow for perturbations where agents' payoff functions do not coincide with those in the unperturbed environment with probability bounded away from zero, there exists such a perturbation under which every equilibrium of the game
implements a social choice function that is bounded away from the desired one $f$. The distance between every implemented outcome and $f$ is bounded from below by a linear function of $\eta$, with the coefficient depending only on the social choice function $f$.
For example,
if $f(\theta)$ is a pure outcome for every $\theta \in \Theta$, then $k(f)$ equals $1$.
This result together with our previous results implies that
the perturbed environment being \textit{close} to the unperturbed environment
is somewhat necessary
for robust implementation.
\begin{proof}[Proof of Proposition 1:]
For any finite mechanism $\mathcal{M} \equiv \{M_1,M_2,g,t_1,t_2\}$, let
\begin{equation*}
    X(\mathcal{M}) \equiv \max_{(i, m_1,m_2) \in \{1,2\} \times M_1 \times M_2} \Big| t_i(m_1,m_2)\Big|.
\end{equation*}
By definition, $X(\mathcal{M})$ exists and is finite.
Since $f$ is non-constant, there exists $\theta^* \in \Theta$ such that
\begin{equation*}
    f(\theta^*) \notin \underbrace{\textrm{co} \Big(
    \{f(\theta)\}_{\theta \in \Theta} \backslash \{f(\theta^*)\}
    \Big)}_{ \equiv \mathcal{Y}}.
\end{equation*}
According to the separating hyperplane theorem, there exists $v: Y \rightarrow \mathbb{R}$ such that $v(f(\theta^*)) < \min_{y \in \mathcal{Y}} v(y)$. Hence, there exists $C>0$ such that $ \Big( \min_{y \in \mathcal{Y}} v(y)-v(f(\theta^*))\Big) C > 4X(\mathcal{M})$.

First, consider a perturbation $\mathcal{G}^+$ where $\widetilde{u}_1(\omega,\theta,y)=C v(y)$ for all $\omega \in \Omega$. If $\mathcal{M}$ implements $f(\theta^*)$ in state $\theta^*$, then there exists $m_2^* \in \Delta (M_2)$ such that
  \begin{equation}\label{4.1}
 \max_{m_1 \in \Delta(M_1)} \{C v(g(m_1,m_2^*)) +t_1(m_1,m_2^*)\} \leq \underbrace{C v(f(\theta^*))+X (\mathcal{M})}_{\textrm{agent 1's highest payoff when the outcome belongs to $\mathcal{Y}$}}.
  \end{equation}

Next, consider another perturbation $\mathcal{G}^-$ where $\widetilde{u}_2(\omega,\theta,y)=-C v(y)$ for all $\omega \in \Omega$. According to (\ref{4.1}), agent 2's payoff is at least
\begin{equation}\label{4.2}
    \min_{m_1 \in \Delta (M_1)} \{-C v(g(m_1,m_2^*)) +t_2(m_1,m_2^*)\}
\end{equation}
if he plays according to $m_2^*$. Since $C>0$ is chosen such that $\Big( \min_{y \in \mathcal{Y}} v(y)-v(f(\theta^*))\Big) C > 4X(\mathcal{M})$ and $X(\mathcal{M}) \geq |t_i(m_1,m_2)|$ for every $i$ and $(m_1,m_2)$, inequalities (\ref{4.1}) and (\ref{4.2}) imply that
\begin{equation}\label{4.3}
     \min_{m_1 \in \Delta (M_1)} \{-C v(g(m_1,m_2^*)) +t_2(m_1,m_2^*)\} \geq -C v(f(\theta^*))-3X (\mathcal{M})> \max_{y \in \mathcal{Y}} \{-C v(y) \}+X(\mathcal{M}).
\end{equation}

For an outcome in $\mathcal{Y}$ to be implemented in any state under perturbation $\mathcal{G}^-$, it must be the case that agent 2's payoff is no more than $\max_{y \in \mathcal{Y}} \{-C v(y)\}+X(\mathcal{M})$. Inequality (\ref{4.3}) implies that no outcome in $\mathcal{Y}$ can be implemented in any state, which implies that every mechanism $\mathcal{M}$ that can implement $f$ under perturbation $\mathcal{G}^+$ cannot implement $f$ under perturbation $\mathcal{G}^-$.
\end{proof}

\subsection{Full Implementation}\label{sub5.2}
Our main results in Section \ref{sec4} focus on robust \textit{partial} implementation: The planner's objective is to design a mechanism so that for every small perturbation, there exists \textit{one} equilibrium whose induced outcome is close to $f$. In what follows, we discuss whether the planner can design a mechanism that can fully or virtually implement $f$ under every small perturbation.

Formally, we say that $f$ is \textit{fully implementable} if there exists a finite mechanism $\mathcal{M} \equiv \{M_1,M_2,g,t_1,t_2\}$ such that $g_{\sigma}(\theta)=f(\theta)$ for every $\theta \in \Theta$ and every equilibrium $\sigma$ under mechanism $\mathcal{M}$.
We say that $f$ is \textit{virtually implementable} if for every $\varepsilon>0$, there exists a finite mechanism $\mathcal{M}$ such that
$||g_{\sigma}(\theta)-f(\theta)||_{TV} \leq \varepsilon$ for every $\theta \in \Theta$ and every equilibrium $\sigma$ under mechanism $\mathcal{M}$.
We say that $f$ is \textit{strongly-robust implementable} if for every $\varepsilon>0$, there exists $\eta>0$ such that for every $\eta$-perturbation $\mathcal{G}$, $||g_{\sigma(\mathcal{G})}(\theta)-f(\theta)||_{TV} \leq \varepsilon$ for every $\theta \in \Theta$ and every equilibrium $\sigma(\mathcal{G})$ of $(\mathcal{M},\mathcal{G})$.
We introduce two conditions under which full and virtual implementation is impossible.
\begin{Proposition}\label{Prop2}
Suppose $f$ is non-constant.
\begin{enumerate}
  \item If $(u_1,u_2)$ do not depend on $\theta$ and $c_1,c_2>0$, then $f$ is not virtually implementable.
  \item For every $(u_1,u_2)$, there exists $\overline{c}>0$ such that $f$ is not virtually implementable when $c_1,c_2 > \overline{c}$.
\end{enumerate}
\end{Proposition}
\begin{proof}[Proof of Proposition 2:]
When $(u_1,u_2)$ do not depend on $\theta$ and $c_1,c_2>0$, there always exists an equilibrium where neither agent pays the strictly positive cost to learn the state. In this equilibrium, the implemented outcome does not depend on the state, which implies that no mechanism can virtually implement any non-constant $f$.

Next, we show that no mechanism can virtually implement non-constant $f$ when $c_1$ and $c_2$ are sufficiently large. For every $u_1$ and $u_2$, let
\begin{equation*}
    X(u_1,u_2) \equiv \max_{i \in \{1,2\}} \Big|\max_{\theta,y} u_i(\theta,y)-\min_{\theta,y} u_i(\theta,y)\Big|.
\end{equation*}
Fix any finite mechanism $\mathcal{M}$, for every $m_2 \in \Delta (M_2)$, let
\begin{equation*}
    T(m_2) \equiv \max_{m_1 \in M_1} t_1(m_1,m_2).
\end{equation*}
Suppose agent $1$ believes that agent 2's message is $m_2$, the difference between his expected payoff when he learns $\theta$ and when he does not learn $\theta$ is
\begin{equation}\label{5.0}
\mathbb{E}\Big[
\max_{m_1 \in M_1} \{u_1(\theta,g(m_1,m_2))+t_1(m_1,m_2)\}
\Big]- \max_{m_1 \in M_1} \mathbb{E}\Big[u_1(\theta,g(m_1,m_2))+t_1(m_1,m_2)\Big].
\end{equation}
By definition, if $m_1^* \in \arg\max_{m_1 \in M_1} \mathbb{E}\Big[u_1(\theta,g(m_1,m_2))+t_1(m_1,m_2)\Big]$, then $t_1(m_1^*,m_2) \geq T(m_2)-X(u_1,u_2)$. This implies that the value of  (\ref{5.0}) is no more than $2X (u_1,u_2)$, and therefore, agent 1 has no incentive to learn $\theta$ when $c_1 > 2X(u_1,u_2)$. In addition, when agent 1 believes that agent 2's message is $m_2$, sending a message that belongs to $\arg\max_{m_1 \in M_1} \mathbb{E}\Big[u_1(\theta,g(m_1,m_2))+t_1(m_1,m_2)\Big]$ regardless of the state is one of agent 1's best replies.

Similarly, suppose $c_2>2X (u_1,u_2)$. For every $m_1 \in \Delta (M_2)$, when agent $2$ believes that agent 1's message is $m_1$, sending a message that belongs to $\arg\max_{m_2 \in M_2} \mathbb{E}\Big[u_2(\theta,g(m_1,m_2))+t_2(m_1,m_2)\Big]$ regardless of the state is one of agent 2's best replies. For every $\mathcal{M}$, consider an auxiliary two-player normal form game where agent $i \in \{1,2\}$ has a finite set of pure strategies $M_i$ and his payoff is
\begin{equation*}
    \mathbb{E}\Big[
    u_i(\theta,g(m_1,m_2))+t_i(m_1,m_2)
    \Big]
\end{equation*}
when he uses strategy $m_i$ and his opponent uses strategy $m_{-i}$. Since this is a finite game, a Nash equilibrium $(m_1,m_2) \in \Delta (M_1) \times \Delta (M_2)$ exists. By construction, agent 1 sending $m_1$ regardless of $\theta$ and agent 2 sending $m_2$ regardless of $\theta$ is an equilibrium under mechanism $\mathcal{M}$. In this equilibrium, the implemented outcome is the same for all $\theta$, which means that it cannot fully implement $f$ when $f$ is non-constant.
\end{proof}

Proposition \ref{Prop2} and its proof imply that when agents' costs of learning are large relative to the responsiveness of $(u_1,u_2)$ with respect to $\theta$, there always exists an equilibrium where neither agent pays the cost to learn the state regardless of how large $(t_1,t_2)$ are. Intuitively, each agent's transfer depends only on the message profile but not on the realized state. When agent 2's message does not depend on the state, the only incentive for agent 1 to learn the state is to induce a more favorable state-contingent outcome. When his cost of learning outweighs the benefit from learning, he has no incentive to learn the state, which gives rise to equilibria where no agent learns the state. By contrast, our main results focus on partial implementation and in the equilibria we construct, both agents' messages depend on the state. As a result, every agent's benefit from learning not only comes from inducing a better state-contingent outcome, but also comes from his incentive to receive a higher transfer. As a result, the planner can robustly implement any state-contingent $f$ even when $c_1$ and $c_2$ are large and $(u_1,u_2)$ do not depend on the state.

Nevertheless, $f$ is fully implementable when one of the agent's payoff function satisfies a strict version of Rochet (1987)'s cyclical monotonicity condition and that agent's cost of learning is sufficiently small.
Formally, for every $f: \Theta \rightarrow \Delta (Y)$ and $u_i: \Theta \times Y \rightarrow \mathbb{R}$, we say that $u_i$ and $f$ satisfy \textit{strict cyclical monotonicity} if for every permutation $\tau$ of $\Theta$,
  \begin{equation}\label{5.1}
    \sum_{\theta \in \Theta} u_i(\theta,f(\theta)) \geq \sum_{\theta \in \Theta} u_i \Big( \theta , f(\tau(\theta)) \Big).
  \end{equation}
and for every permutation $\tau$ such that $f(\tau(\theta)) \neq f(\theta)$ for some $\theta \in \Theta$,
  \begin{equation}\label{5.2}
    \sum_{\theta \in \Theta} u_i(\theta,f(\theta)) > \sum_{\theta \in \Theta} u_i \Big( \theta , f(\tau(\theta)) \Big).
  \end{equation}
Condition (\ref{5.1}) is the cyclical monotonicity condition in Rochet (1987). Condition (\ref{5.2}) is novel: for every $f$ that depends nontrivially on $\theta$, this condition rules out $u_i$ that does not depend on $\theta$.
\begin{Proposition}\label{Prop3}
If there exists $i \in \{1,2\}$ such that $u_i$ and $f$ satisfy strict cyclical monotonicity,
then there exists $\overline{c}>0$ such that for every $c< \overline{c}$, there exists a finite mechanism $\mathcal{M}$ that can fully implement $f$
and can strongly-robust implement $f$.
\end{Proposition}
\begin{proof}[Proof of Proposition 3:]
If $f$ is constant, then the result is straightforward. We focus on $f$ that is non-constant.
Consider a mechanism where $M_i=\Theta$, $M_{-i}=\{1\}$, $g(m_i,m_{-i})=f(m_i)$, $t_i(m_1,m_2)$ depends only on $m_1$, and $t_{-i}(m_1,m_2)=0$. Since $f$ and $u_i$ satisfy strict cyclical monotonicity, there exists $t_i : \Theta \rightarrow \mathbb{R}$ such that
\begin{enumerate}
  \item $t_i(\theta)=t_i(\theta')$ for every $\theta,\theta' \in \Theta$ such that $f(\theta) = f(\theta')$,
  \item $u_i(\theta,f(\theta))+t_i(\theta) > u_i(\theta,f(\theta'))+t_i(\theta')$ for every $\theta,\theta' \in \Theta$ such that $f(\theta) \neq f(\theta')$.
\end{enumerate}
Under such a mechanism, agent $i$ chooses an outcome in $\{f(\theta)\}_{\theta \in \Theta}$ and receive an additional reward $t_i(\theta)$ for implementing $f(\theta)$. Agent $i$ has a strict incentive to choose $f(\theta)$ in state $\theta$, so he has a strict incentive to learn the state when $c$ is small enough. Under every $\eta$-perturbation $\mathcal{G}$, every normal type of agent $i$ has a strict incentive to learn the state and to induce outcome $f(\theta)$ in state $\theta$. This completes the proof.
\end{proof}

\section{Conclusion}\label{sec6}
We examine the problem faced by a planner when he wants to robustly implement a state-contingent social choice function when (i) agents need to incur costs to learn the state, (ii) the planner faces uncertainty about agents' costs of obtaining information, their biases over outcomes, as well as their beliefs and higher-order-beliefs about each other's payoffs. We introduce mechanisms that robustly implement the desired social choice function when the state distribution satisfies a generic assumption, or when the planner knows an upper bound on agents' costs of obtaining information.
We conclude by discussing the related literature.

Our work contributes to the literature on robust implementation, pioneered by Bergemann and Morris (2005). Our paper is closely related to Oury and Tercieux (2012), who take an interim perspective and require the outcome induced by every nearby type to be close to that induced by the original type. They show that Maskin monotonicity is necessary for robust partial implementation.\footnote{The literature on full implementation (Maskin 1999) and virtual implementation (Abreu and Matsushima 1992) examine whether the planner can implement or approximately implement the desired social choice function in all equilibria. Unlike our model, they do not consider any perturbations and only require full or virtual implementation in the unperturbed environment. Hence, our result is neither stronger nor weaker than theirs.} By contrast, we take an ex ante perspective by requiring that the desired outcome to be implemented with probability close to one in all nearby type spaces. We show that all social choice functions are robustly implementable and our mechanisms are robust to small trembles and noisy information about the state.
Our results design mechanisms that can robustly elicit costly information when the principal faces
uncertainty about agents' motives, costs of obtaining information, as well as beliefs and higher-order-beliefs
while assuming that agents' information acquisition technologies are known and their signals about the state are highly correlated.
This stands in contrast to Carroll (2019) who focuses on robust contracting when the principal faces uncertainty about the agent's information acquisition technology and
Aghion, Fudenberg, Holden, Kunimoto and Tercieux (2012) that perturbs players' information about the state.

Our work is related to the literature on robust prediction in games, pioneered by Rubinstein (1989). Our notion of robustness builds on the notion of robust equilibrium in Kajii and Morris (1997). Their notion of has been broadly applied to study the robustness of equilibria in potential games (Ui 2001, Morris and Ui 2005) and supermodular games (Oyama and Takahashi 2020). The key difference is that in our model, agents' payoffs in the perturbed game do not directly depend on their actions, which are their messages in our mechanism design setting. The assumption that agents' messages are cheap talk is common in the robust mechanism design literature, such as Bergemann and Morris (2005), Aghion, Fudenberg, Holden, Kunimoto and Tercieux (2012).

Finally, our work is also related to the large literature on contracting for information acquisition such as Zermeno (2011) and Clark and Reggiani (2021), as well as mechanism design with costly information acquisition such as Cr\'{e}mer and Khalil (1992), Persico (2000), Bergemann and V\"{a}lim\"{a}ki (2002), and Li (2019). Those papers that study agents' incentives to acquire information and/or characterize the optimal mechanism under a fixed informational environment. In contrast, we examine whether it is possible to (approximately) implement a desired social choice function in \textit{all} nearby environments. A notable exception is Carroll (2019), who characterizes the optimal contract for information acquisition when
the principal faces uncertainty about the set of information acquisition technologies available to the agent, but can condition the agent's transfer on the realized state. By contrast, the state cannot be verified ex post in our setting, and the principal faces uncertainty not about the information acquisition technology, but about the agents' biases over outcomes, costs of obtaining information, and beliefs about each other's payoffs.

\end{spacing}

\end{document}